\documentclass[11pt]{article}

\usepackage{fullpage}
\usepackage{amsmath}
\usepackage{amssymb}
\usepackage{amsthm}

\DeclareMathOperator*{\argmin}{argmin}

\newcommand{\minus}[1]{\setminus\{#1\}}

\newtheorem{theorem}{Theorem}

\newtheorem{remark}[theorem]{Remark}

\begin{document}

\title{Similar Elements and Metric Labeling on Complete Graphs}
\author{Pedro F. Felzenszwalb \\
  Brown University \\
  Providence, RI, USA \\
  {\tt pff@brown.edu}} 
\maketitle

We consider a problem that involves finding similar elements in a
collection of sets.  The problem is motivated by applications in
machine learning and pattern recognition (see, e.g. \cite{MR98}).
Intuitively we would like to discover something in common among a
collection of sets, even when the sets have empty intersection.  A
solution involves selecting an element from each set such that the
selected elements are close to each other under an appropriate metric.
We formulate an optimization problem that captures this notion and
give an efficient approximation algorithm that finds a solution within
a factor of 2 of the optimal solution.

The similar elements problem is a special case of the metric labeling
problem defined in \cite{KT02} and we also give an efficient
2-approximation algorithm for the metric labeling problem on
\emph{complete graphs}.  Metric labeling on complete graphs
generalizes the similar elements problem to include costs for selecting
elements in each set.

The algorithms described here are similar to the ``center star'' method
for multiple sequence alignment described in \cite{G93}.

Beyond producing solutions with good theoretical guarantees, the
algorithms described here are also practical.  A version
of the algorithm for the similar elements problem has been implemented
and used to find objects in a collection of photographs
\cite{SS16}.

\section{Similar Elements}

Let $X$ be a (possibly infinite) set and $d$ be a metric on $X$.  Let
$S_1,\ldots,S_n$ be $n$ finite subsets of $X$.  The goal of the
\emph{similar elements} problem is to select an element from each set $S_i$
such that the selected elements are close to each other under the
metric $d$.  One motivation is for discovering something in common
among the sets $S_1,\ldots,S_n$ even when they have empty intersection.

We formalize the problem as the minimization of the sum
of pairwise distances among selected elements.  Let $x=(x_1,\ldots,x_n)$
with $x_i \in S_i$.  Define the similar elements objective as,
\begin{equation}
  c(x) = \sum_{1 \le i,j \le n} d(x_i,x_j).
\end{equation}
Let $x^* = \argmin_x c(x)$ be an optimal solution for the similar
elements problem.

Optimizing $c(x)$ appears to be difficult, but we can define easier
problems if we ignore some of the pairwise distances in the objective.
In particular we define $n$ different ``star-graph'' objective
functions as follows.  For each $1 \le r \le n$ define the objective
$c^r(x)$ to account only for the terms in $c(x)$ involving $x_r$,
\begin{equation}
  c^r(x) = \sum_{j \neq r} d(x_r,x_j).
\end{equation}

Let $x^r = \argmin_x c^r(x)$ be an optimal solution for the
optimization problem defined by $c^r(x)$.  We can compute $x^r$
efficiently using a simple form of dynamic programming, by first
computing $x^r_r$ and then computing $x^r_j$ for $j \neq r$.

\begin{equation}
  x^r_r = \argmin_{x_r \in S_r} \sum_{j \neq r} \min_{x_j \in S_j} d(x_r,x_j),
  \label{eqn:DP1similar}
\end{equation}
\begin{equation}
  x^r_j = \argmin_{x_j \in S_j} d(x^r_r,x_j).
  \label{eqn:DP2similar}
\end{equation}

Each of the $n$ ``star-graph'' objective functions leads to a possible
solution.  We then select from among the solutions $x^1,\ldots,x^n$
as follows,
\begin{eqnarray}
  \hat{r} & = & \argmin_{1 \le r \le n} c^r(x^r), \\
  \hat{x} & = & x^r.
\end{eqnarray}

\begin{theorem}
  The algorithm described above finds a 2-approximate solution for the similar elements
  problem.  That is,
  $$c(\hat{x}) \le 2 c(x^*).$$
\end{theorem}

\begin{proof}
First note that,
$$c(x) = \sum_{r=1}^n c^r(x).$$ 
Since the minimum of a set of values is at most the average, and $x^r$ minimizes $c^r(x)$,
$$\min_{1 \le r \le n} c^r(x^r) \le \frac{1}{n} \sum_{r=1}^n c^r(x^r) \le \frac{1}{n} \sum_{r=1}^n c^r(x^*) = \frac{1}{n} c(x^*).$$
By the triangle inequality we have 
$$c(x) = \sum_{1 \le i,j \le n} d(x_i,x_j) \le \sum_{1 \le i,j \le n} (d(x_i,x_r) + d(x_r,x_j)) = 2n \sum_{l=1}^n d(x_r,x_l) = 2nc^r(x).$$
Therefore
$$c(\hat{x}) \le 2nc^{\hat{r}}(\hat{x}) = 2n \min_{1 \le r \le n} c^r(x^r) \le 2c(x^*).$$
\end{proof}

To analyze the running time of the algorithm we assume
the distances $d(p,q)$ between pairs of elements in $S = S_1 \cup
\cdots \cup S_n$ are either pre-computed and given as part of the
input, or they can each be computed in $O(1)$ time.

Let $k = \max_{1 \le i \le n} |S_i|$.  The first stage of the
algorithm involves $n$ optimization problems that can be solved in
$O(nk^2)$ time each.  The second stage of the algorithm involves
selecting one of the $n$ solutions, and takes $O(n^2)$ time.

\begin{remark}
  If each of the sets $S_1,\ldots,S_n$ has size at most $k$ the
  running time of the approximation algorithm for the similar elements
  problem is $O(n^2k^2)$.
\end{remark}

The bottleneck of the algorithm is the evaluation of the
minimizations over $x_j \in S_j$ in (\ref{eqn:DP1similar}) and
(\ref{eqn:DP2similar}).  This computation is equivalent to a
nearest-neighbor computation, where we want to find a point from a set
$S \subseteq X$ that is closest to a query point $q \in X$.  When the
nearest-neighbor computation can be done efficiently (with an
appropriate data structure) the running time of the similar elements
approximation algorithm can be reduced.

\section{Metric Labeling on Complete Graphs}

Let $G=(V,E)$ be an undirected simple graph on $n$ nodes
$V=\{1,\ldots,n\}$.  Let $L$ be a finite set of labels with $|L|=k$
and $d$ be a metric on $L$.  For $i \in V$ let $m_i$ be a non-negative
function mapping labels to real values.  The \emph{unweighted metric
labeling} problem on $G$ is to find a labeling $x=(x_1,\ldots,x_n) \in
L^n$ minimizing
\begin{equation}
  c(x) = \sum_{i \in V} m_i(x_i) + \sum_{\{i,j\} \in E} d(x_i,x_j).
\end{equation}

Let $x^* = \argmin_x c(x)$.  This optimization problem can be solved
in polynomial time using dynamic programming if $G$ is a tree.  Here
we consider the case when $G$ is the \emph{complete graph} and give an
efficient 2-approximation algorithm based on the solution of several metric
labeling problems on star graphs.

For each $r \in V$ define a different objective function, $c^r(x)$,
corresponding to a metric labeling problem on a star graph with vertex
set $V$ rooted at $r$,
\begin{equation}
  c^r(x) = \sum_{i \in V} \frac{m_i(x_i)}{n} + \sum_{j \in V\minus{r}} \frac{d(x_r,x_j)}{2}.
\end{equation}

Let $x^r=\argmin_x c^r(x)$.
We can solve this optimization problem in $O(nk^2)$ time using a simple form of
dynamic programming.  First compute an
optimal label for the root vertex using one step of dynamic
programming,
\begin{equation}
  x^r_r = \argmin_{x_r \in L} \left( \frac{m_r(x_r)}{n} + \sum_{j \in V\minus{r}}  \min_{x_j \in L} \left( \frac{m_j(x_j)}{n} + \frac{d(x_r,x_j)}{2} \right) \right).
  \label{eqn:DP1metric}
\end{equation}
Then compute $x^r_j$ for $j \in V\minus{r}$,
\begin{equation}
  x^r_j = \argmin_{x_j \in L} \left( \frac{m_j(x_j)}{n} + \frac{d(x^r_r,x_j)}{2} \right).
  \label{eqn:DP2metric}
\end{equation}

Optimizing each $c^r(x)$ separately leads to $n$ possible solutions $x^1,\ldots,x^n$, and we
select one of them as follows,
\begin{eqnarray}
  \hat{r} & = & \argmin_{r \in V} c^r(x^r), \\
  \hat{x} & = & x^r.
\end{eqnarray}

\begin{theorem}
  The algorithm described above finds a 2-approximate solution for the metric
  labeling problem on a complete graph.  That is,
  $$c(\hat{x}) \le 2 c(x^*).$$
\end{theorem}

\begin{proof}
First note that,
$$c(x) = \sum_{r=1}^n c^r(x).$$
Since the minimum of a set of values is at most the average, and $x^r$ minimizes $c^r(x)$,
$$\min_{1 \le r \le n} c^r(x^r) \le \frac{1}{n} \sum_{r=1}^n c^r(x^r) \le \frac{1}{n} \sum_{r=1}^n c^r(x^*) = \frac{1}{n} c(x^*).$$
Since $d$ is a metric and $m_i$ is non-negative,
\begin{eqnarray*}
  c(x) & = & \sum_{i \in V} m_i(x_i) + \sum_{\{i,j\} \in E} d(x_i,x_j) \\
  & = & \sum_{i \in V} m_i(x_i) + \sum_{(i,j) \in V^2} \frac{d(x_i,x_j)}{2} \\
  & \le & \sum_{i \in V} m_i(x_i) + \sum_{(i,j) \in V^2} \left( \frac{d(x_i,x_r)}{2} + \frac{d(x_r,x_j)}{2} \right) \\
  & = & \sum_{i \in V} m_i(x_i) + 2n\sum_{l \in V\minus{r}} \frac{d(x_r,x_l)}{2} \\
  & \le & 2n \sum_{i \in V} \frac{m_i(x_i)}{n} + 2n \sum_{l \in V\minus{r}} \frac{d(x_r,x_l)}{2} \\
  & = & 2nc^r(x).
\end{eqnarray*}
Therefore
$$c(\hat{x}) \le 2nc^{\hat{r}}(\hat{x}) = 2n \min_{1 \le r \le n} c^r(x^r) \le 2c(x^*).$$
\end{proof}

The first stage of the algorithm involves $n$ optimization problems
that can be solved in $O(nk^2)$ time each.  The second stage involves
selecting one of the $n$ solutions, and takes $O(n^2)$ time.

\begin{remark}
  The running time of the approximation algorithm for the metric
  labeling problem on complete graphs is $O(n^2k^2)$.
\end{remark}

\subsection*{Acknowledgments}

We thank Caroline Klivans, Sarah Sachs, Anna Grim, Robert Kleinberg
and Yang Yuan for helpful discussions about the contents of this
report.  This material is based upon work supported by the National
Science Foundation under Grant No. 1447413.

\bibliography{pairs}

\begin{thebibliography}{1}

\bibitem{G93}
Dan Gusfield.
\newblock Efficient methods for multiple sequence alignment with guaranteed
  error bounds.
\newblock {\em Bulletin of Mathematical Biology}, 55(1):141--154, 1993.

\bibitem{KT02}
Jon Kleinberg and Eva Tardos.
\newblock Approximation algorithms for classification problems with pairwise
  relationships: Metric labeling and markov random fields.
\newblock {\em Journal of the ACM}, 49(5):616--639, 2002.

\bibitem{MR98}
Oded Maron and Aparna~Lakshmi Ratan.
\newblock Multiple-instance learning for natural scene classification.
\newblock In {\em International Conference on Machine Learning}, volume~98,
  pages 341--349, 1998.

\bibitem{SS16}
Sarah Sachs.
\newblock Similar-part approximation using invariant feature descriptors.
\newblock Undergraduate Honors Thesis, Brown University, 2016.

\end{thebibliography}
\bibliographystyle{plain}

\end{document}